\documentclass{llncs}

\usepackage{amssymb,amsmath,graphics,color,comment,enumerate}
\usepackage{graphicx}
\usepackage{multicol}
\newcommand{\NP}{{\sf NP}}

\usepackage{pgf,tikz}
\usetikzlibrary{arrows,shapes}

\begin{document}
\title{Narrowing the Complexity Gap for\\ Colouring  $(C_s,P_t)$-Free Graphs\thanks{An extended abstract of this paper appeared in the proceedings of AAIM 2014~\cite{HJP14}.}}
\author{Shenwei Huang\inst{1}, Matthew Johnson\inst{2}\and Dani\"el Paulusma\inst{2}\thanks{Author supported by EPSRC (EP/G043434/1).}
}
\institute{School of Computing Science, Simon Fraser University\\ Burnaby B.C., V5A 1S6, Canada\\
\texttt{shenweih@sfu.ca}
\and
School of Engineering and  Computing Sciences, Durham University,\\
Science Laboratories, South Road,
Durham DH1 3LE, United Kingdom
\texttt{\{matthew.johnson2,daniel.paulusma\}@durham.ac.uk}}
\maketitle
\setcounter{footnote}{0}

\begin{abstract}
For a positive integer $k$ and graph $G=(V,E)$, a {\it $k$-colouring} of $G$ is  a mapping $c: V\rightarrow\{1,2,\ldots,k\}$ such that $c(u)\neq c(v)$ whenever $uv\in E$.
The {\sc $k$-Colouring} problem is to decide, for a given $G$, whether a $k$-colouring of $G$ exists.
The $k$-{\sc Precolouring Extension} problem is to decide, for a given $G=(V,E)$, whether a colouring of a subset of $V$ can be extended to a $k$-colouring of $G$.
A $k$-list assignment of a graph is an allocation of a list --- a subset of  $\{1,\ldots,k\}$ --- to each vertex, and the {\sc List $k$-Colouring} problem is to decide, for a given $G$, whether $G$ has a $k$-colouring in which each vertex is coloured with a colour from its list.
We continued the study of the computational complexity of these three decision problems when restricted to
graphs that do not contain a cycle on $s$ vertices or a path on $t$ vertices as induced subgraphs (for fixed positive integers $s$ and~$t$).
\end{abstract}

\section{Introduction}\label{sec:intro}

Let $G=(V,E)$ be a graph.
A {\it colouring} of $G$ is a mapping $c: V\rightarrow\{1,2,\ldots\}$ such that $c(u)\neq c(v)$ whenever
$uv\in E$. We call $c(u)$  the {\it colour} of $u$. A
{\it $k$-colouring} of $G$ is a colouring with $1\leq c(u)\leq k$ for all $u\in V$.
We study the following decision problem:

\medskip
\noindent \textsc{$k$-Colouring}\\
\mbox{}\rlap{\textit{Instance}\,: }\hphantom{\textit{Question}\,: }A graph $G$.\\
\textit{Question}\,: Is~$G$ $k$-colourable?

\medskip

It is well-known that {\sc $k$-Colouring} is \NP-complete even if $k=3$~\cite{Lo73}, and so the problem  has been studied for special graph classes; see the surveys of Randerath and Schiermeyer~\cite{RS04b} and Tuza~\cite{Tu97}, and the very recent survey of
Golovach, Johnson, Paulusma and Song~\cite{GJPS}.
In this paper, we consider graph classes defined in terms of forbidden induced subgraphs, and study the computational complexity of {\sc $k$-Colouring} and some related problems that we introduce now before stating our results.

A {\it $k$-precolouring} of $G=(V,E)$ is a mapping $c_W :W\rightarrow\{1,2,\ldots k\}$ for some subset $W\subseteq V$.
A $k$-colouring $c$ is an {\it extension} of $c_W$ if $c(v)=c_W(v)$ for each $v \in W$.

\medskip\noindent
\noindent \textsc{$k$-Precolouring Extension}\\
\mbox{}\rlap{\textit{Instance}\,: }\hphantom{\textit{Question}\,: }A graph $G$ and a $k$-precolouring $c_W$ of $G$.\\
\textit{Question}\,: Can $c_W$ be extended to a $k$-colouring of $G$?

\medskip

\noindent
A {\it list assignment} of a graph $G=(V,E)$ is a function $L$ that assigns
a list $L(u)$ of   {\it admissible} colours to each $u\in V$.
If $L(u)\subseteq \{1,\ldots,k\}$ for each $u\in V$, then $L$ is also called a \emph{$k$-list assignment}.
A colouring $c$
{\it respects} $L$ if  $c(u)\in L(u)$ for all $u\in V$.  Here is our next decision problem:

\medskip\noindent
\noindent \textsc{List $k$-Colouring}\\
\mbox{}\rlap{\textit{Instance}\,: }\hphantom{\textit{Question}\,: }A graph $G$ and a $k$-list assignment $L$ for $G$.\\
\textit{Question}\,: Is there a colouring of $G$ that respects $L$?

\medskip

\noindent Note that $k$-{\sc Colouring} can be viewed as a special case of $k$-{\sc Precolouring Extension} which is, in turn, a special case of {\sc List $k$-Colouring}.

A graph is $(C_s, P_t)$-free if it has no subgraph isomorphic to either $C_s$, the cycle on $s$ vertices, or $P_t$, the path on $t$ vertices.
Several papers~\cite{CMSZ14,GPS11,HH13} have considered the computational complexity of our three decision problems when restricted to
$(C_s,P_t)$-free graphs. In this paper, we continue this investigation. Our first contribution is to state the following theorem that provides
a complete summary of our current knowledge.
The cases marked with an asterisk are new results presented in this paper.   We use {\it p-time} to mean polynomial-time throughout the paper.

\begin{theorem}\label{t-coloringall}
Let $k,s,t$ be three positive integers.
The following statements hold for $(C_s,P_t)$-free graphs.
\begin{itemize}
\item[(i)] {\sc List $k$-Colouring} is \NP-complete if\\[-8pt]
\begin{enumerate}
\item[1.$^*$] $k\geq 4$,   $s=3$ and $t\geq 8$
\item[2.{$^*$}]  $k\geq 4$,  $s\geq 5$ and $t\geq 6$.\\[-8pt]
\end{enumerate}
\item[] {\sc List $k$-Colouring} is p-time solvable if\\[-8pt]
\begin{enumerate}
  \setcounter{enumi}{2}
\item[3.\phantom{$^*$}]   $k\leq 2$, $s\geq 3$ and $t\geq 1$
\item[4.\phantom{$^*$}]   $k=3$, $s=3$ and $t\leq 6$
\item[5.\phantom{$^*$}]   $k=3$, $s=4$ and $t\geq 1$
\item[6.\phantom{$^*$}]   $k=3$, $s\geq 5$ and $t\leq 6$
\item[7.\phantom{$^*$}]  $k\geq 4$, $s=3$ and $t\leq 6$
\item[8.\phantom{$^*$}]  $k\geq 4$, $s=4$ and $t\geq 1$
\item[9.\phantom{$^*$}]  $k\geq 4$, $s\geq 5$ and $t\leq 5$.\\[-25pt]
\item[] 
\end{enumerate}
\end{itemize}
\begin{itemize}
\item[(ii)] {\sc $k$-Precolouring Extension} is \NP-complete if\\[-8pt]
\begin{enumerate}
\item[1.\phantom{$^*$}] $k=4$,   $s=3$ and $t\geq 10$
\item[2.\phantom{$^*$}] $k=4$,  $s=5$ and $t\geq 7$
\item[3.\phantom{$^*$}] $k=4$, $s=6$ and $t\geq 7$
\item[4.{$^*$}]  $k=4$, $s=7$ and $t\geq 8$
\item[5.\phantom{$^*$}] $k=4$, $s\geq 8$ and $t\geq 7$
\item[6.\phantom{$^*$}] $k\geq 5$,   $s=3$ and $t\geq 10$
\item[7.{$^*$}]  $k\geq 5$,  $s\geq 5$ and $t\geq 6$.\\[-17pt]
\item[] 
\end{enumerate}
\item[] {\sc $k$-Precolouring Extension} is p-time solvable if\\[-8pt]
\begin{enumerate}
  \setcounter{enumi}{7}
\item[8.\phantom{$^*$}]  $k\leq 2$, $s\geq 3$ and $t\geq 1$
\item[9.\phantom{$^*$}]   $k=3$, $s=3$ and $t\leq 6$
\item[10.\phantom{$^*$}]   $k=3$, $s=4$ and $t\geq 1$
\item[11.\phantom{$^*$}]   $k=3$, $s\geq 5$ and $t\leq 6$
\item[12.\phantom{$^*$}]  $k\geq 4$, $s=3$ and $t\leq 6$
\item[13.\phantom{$^*$}]  $k\geq 4$, $s=4$ and $t\geq 1$
\item[14.\phantom{$^*$}]  $k\geq 4$, $s\geq 5$ and $t\leq 5$.\\[-25pt]
\item[] 
\end{enumerate}
\end{itemize}
\begin{itemize}
\item[(iii)] $k$-{\sc Colouring} is \NP-complete if\\[-8pt]
\begin{enumerate}
\item[1.{$^*$}]  $k=4$, $s=3$ and $t\geq 22$
\item[2.\phantom{$^*$}]  $k=4$,  $s=5$ and $t\geq 7$
\item[3.\phantom{$^*$}]  $k=4$, $s=6$ and $t\geq 7$
\item[4.\phantom{$^*$}]  $k=4$, $s=7$ and $t\geq 9$
\item[5.\phantom{$^*$}]  $k=4$, $s\geq 8$ and $t\geq 7$
\item[6.{$^*$}]  $k\geq 5$, $s=3$ and $t\geq t_k$ where $t_k$ is a constant that only depends on $k$
\item[7.\phantom{$^*$}]  $k\geq 5$,  $s=5$ and $t\geq 7$
\item[8.\phantom{$^*$}]  $k\geq 5$,  $s\geq 6$ and $t\geq 6$.\\[-17pt]
\item[] 
\end{enumerate}
\item[] $k$-{\sc Colouring} is p-time solvable if\\[-8pt]
\begin{enumerate}
  \setcounter{enumi}{7}
\item[9.\phantom{$^*$}]   $k\leq 2$, $s\geq 3$ and $t\geq 1$
\item[10.\phantom{$^*$}]   $k=3$, $s=3$ and $t\leq 7$
\item[11.\phantom{$^*$}]   $k=3$, $s=4$ and $t\geq 1$
\item[12.\phantom{$^*$}]   $k=3$, $s\geq 5$ and $t\leq 7$
\item[13.\phantom{$^*$}]  $k=4$, $s=3$ and $t\leq 6$
\item[14.\phantom{$^*$}]  $k=4$, $s=4$ and $t\geq 1$
\item[15.\phantom{$^*$}]  $k=4$, $s=5$ and $t\leq 6$
\item[16.\phantom{$^*$}]  $k=4$, $s\geq 6$ and $t\leq 5$
\item[17.\phantom{$^*$}]   $k\geq 5$, $s=3$ and $t\leq k+2$
\item[18.\phantom{$^*$}]   $k\geq 5$, $s=4$ and $t\geq 1$
\item[19.\phantom{$^*$}]   $k\geq 5$, $s\geq 5$ and $t\leq 5$.
\item[] 
\end{enumerate}
\end{itemize}
\end{theorem}

The new results on {\sc List $k$-Colouring}, {\sc $k$-Precolouring Extension} and {\sc $k$-Colouring} are in Sections~\ref{s-4col}, \ref{s-precol} and~\ref{s-colouring} respectively. 
In these three sections we often prove stronger statements, for example, on (chordal) bipartite graphs, to strengthen existing results in the literature as much as we can.
In Section~\ref{s-summary}, we prove Theorem~\ref{t-coloringall} by combining a number of previously known results with our new results, and
in Section~\ref{s-open} we summarize the open cases and pose a number of related open problems.

We introduce some more terminology that we will need.  Let $G=(V,E)$ be a graph.
The {\it chromatic number} of $G$ is the smallest integer $k$ for which $G$ has a $k$-colouring.
Let $\{H_1,\ldots,H_p\}$ be a set of graphs. We say that
$G$ is {\it $(H_1,\ldots,H_p)$-free} if $G$ has no induced subgraph isomorphic to a graph in $\{H_1,\ldots,H_p\}$; if $p=1$, we  write $H_1$-free instead of $(H_1)$-free.
The {\it complement} of $G$, denoted by $\overline{G}$, has vertex set $V$ and an edge between two distinct vertices
if and only if these vertices are not adjacent in $G$.
The disjoint union of two graphs $G$ and $H$ is denoted $G+H$, and
the disjoint union of $r$ copies of $G$ is denoted $rG$.
The {\it girth} of  $G$ is the number of vertices of a shortest cycle in~$G$ or infinite if~$G$ has no cycle. 
Note that a graph has girth at least $g$ for some $g\geq 4$ if and only if it is $(C_3,\ldots,C_{g-1})$-free.
To add a pendant vertex to a vertex $u \in V$, means to obtain a new graph from $G$ by adding one more vertex and making it adjacent only to $u$.
We denote the complete graph on $r$ vertices by $K_r$.
A graph is {\it chordal bipartite} if it is bipartite and every induced cycle has exactly four vertices.

We complete this section by providing some context for our work on $(C_s, P_t)$-free graphs.
We comment that it can be seen as a natural continuation of  investigations into the complexity of {\sc $k$-Colouring} and {\sc List $k$-Colouring} for $P_t$-free graphs
(see~\cite{GJPS}).
The sharpest results are the following.
Ho\`ang et al.~\cite{HKLSS10} proved that, for all $k\geq 1$, {\sc List $k$-Colouring} is p-time solvable on $P_5$-free graphs.
Huang~\cite{Hu13} proved
that $4$-{\sc Colouring} is \NP-complete for $P_7$-free graphs and that $5$-{\sc Colouring} is \NP-complete for $P_6$-free graphs.
Recently, Chudnovsky, Maceli and Zhong~\cite{CMZ14a,CMZ14b}
announced a p-time algorithm for solving {\sc 3-Colouring} on $P_7$-free graphs.
Broersma et al.~\cite{BFGP13} proved that {\sc List 3-Colouring} is p-time solvable for $P_6$-free graphs.
Golovach, Paulusma and Song~\cite{GPS12} proved that {\sc List 4-Colouring} is \NP-complete for $P_6$-free graphs.
These results lead to the following table (in which the open cases are denoted by ``?'').

\begin{table}[h]
\begin{center}
\resizebox{360pt}{35pt}{
\begin{tabular}{c|c|c|c|c||c|c|c|c||c|c|c|c}
& \multicolumn{4}{c||}{{\sc $k$-Colouring}} & \multicolumn{4}{c||}{ {\sc $k$-Precolouring Extension}}&\multicolumn{4}{c}{ {\sc List $k$-Colouring}}\\
\cline{2-13}
                                 & {\small $k=3$}             & {\small $k=4$}                  & {\small $k=5$}               & {\small $k\ge 6$}                          & {\small $k=3$}             & {\small $k=4$}                  & {\small $k=5$}               & {\small $k\ge 6$}                             & {\small $k=3$}             & {\small $k=4$}                  & {\small $k=5$}               & {\small $k\ge 6$} \\
\hline
{\small $t\leq 5$}    & {\small P} & {\small P}      & {\small P} & {\small P}  & {\small P} & {\small P}      & {\small P} & {\small P}
 & {\small P} & {\small P}      & {\small P} & {\small P} \\
{\small $t=6$}         & {\small P} & ?                     & {\small NP-c}               & {\small NP-c}                 & {\small P} & ?                     & {\small NP-c}            & {\small NP-c}
&{\small P} & {\small NP-c}                  & {\small NP-c}            & {\small NP-c} \\
{\small $t=7$}         &
{\small P}
& {\small NP-c}                     & {\small NP-c}               & {\small NP-c}  & ?               &  {\small NP-c}    & {\small NP-c}        & {\small NP-c}
& ?               & {\small NP-c}    & {\small NP-c}        & {\small NP-c}\\
{\small $t\geq 8$}   & ?              &  {\small NP-c}    &{\small NP-c} & {\small NP-c}  &? &  {\small NP-c}    &{\small NP-c} & {\small NP-c}
  &?              & {\small NP-c} & {\small NP-c} & {\small NP-c}
\end{tabular}
}
\end{center}
\caption{The complexity of $k$-{\sc Colouring}, $k$-{\sc Precolouring Extension}  and {\sc List $k$-Colouring}
for $P_t$-free graphs.}\label{t-table1}
\end{table}

 \section{New Results for List Colouring}\label{s-4col}

 In this section we give two results on {\sc List $4$-Colouring}.

We first prove that {\sc List $4$-Colouring}  is \NP-complete for the class of $(C_5,C_6,K_4,\overline{P_1+2P_2},\overline{P_1+P_4})$-free graphs.
(We observe that $\overline{P_1+2P_2}$ is also known as the 5-vertex wheel and $\overline{P_1+P_4}$ is sometimes called
the gem or the 5-vertex fan.)
This result strengthens an analogous result on {\sc List $4$-Colouring} of $P_6$-free graphs~\cite{GPS12}, and is obtained by a closer analysis of the hardness reduction used in the proof of that result.  The reduction is from the problem
{\sc Not-All-Equal 3-Sat} with positive literals only which was shown to be \NP-complete by Schaefer~\cite{Schaefer78} and is defined as follows. The input $I$ consists of a set  $X= \{x_1,x_2,\ldots,x_n\}$ of variables, and a set ${\cal C} = \{D_1, D_2, \ldots, D_m\}$ of 3-literal clauses over $X$ in which all literals are positive.
The question is whether there exists a truth assignment for $X$ such that each $D_i$ contains at least one true literal and at least one false literal.
We describe a graph~$J_I$ and 4-list assignment $L$ that are defined using the instance $I$:

\begin{itemize}
\item [$\bullet$]
$a$-type and $b$-type vertices:
for each clause $D_j$, $J_I$ contains two \emph{clause components} $D_j$ and $D_j'$ each isomorphic to $P_5$.  Considered along the paths the vertices in $D_j$ are  $a_{j,1},b_{j,1},a_{j,2},b_{j,2},a_{j,3}$ with
 lists of admissible colours $\{2,4\}, \{3,4\}, \{2,3,4\}, \{3,4\}, \{2,3\}$,  respectively, and the vertices in $D_j'$ are $a'_{j,1},b'_{j,1},a'_{j,2},b'_{j,2},a'_{j,3}$ with lists of admissible colours $\{1,4\}$, $\{3,4\}$, $\{1,3,4\}$, $\{3,4\}$, $\{1,3\}$, respectively.

\item [$\bullet$]
$x$-type vertices: for each variable $x_i$, $J_I$ contains a vertex $x_i$ with list of admissible colours $\{1,2\}$.

\item [$\bullet$]
For every clause $D_j$,
its variables $x_{i_1},x_{i_2},x_{i_3}$ are ordered in an arbitrary (but fixed) way, and in $J_I$ there are
edges $a_{j,h}x_{i_h}$ and $a_{j,h}'x_{i_h}$ for $h=1,2,3$.

\item [$\bullet$]
There is an edge in $J_I$ from every $x$-type vertex to every $b$-type vertex.
\end{itemize}

\noindent
\noindent
See Figure~\ref{fig1} for an example of the graph $J_I$.
In this figure, $D_j$ is a clause with ordered variables $x_{i_1},x_{i_2},x_{i_3}$. The thick edges indicate the
connection between these vertices and the $a$-type vertices of
 the two copies of the clause gadget. Indices from
the labels of the clause gadget vertices have been omitted to aid clarity.

\tikzstyle{vertex}=[circle,draw=black, fill=black, minimum size=5pt, inner sep=1pt]
\tikzstyle{edge} =[draw,-,black,>=triangle 90]

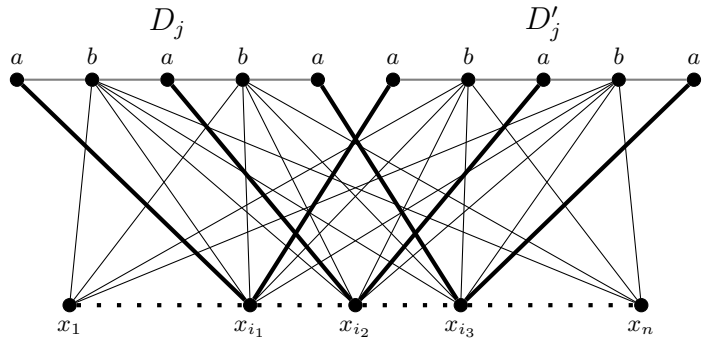
\begin{figure}
\begin{center}
\begin{tikzpicture}[scale=1]

   \foreach \pos/\name / \label in {{(0,1)/a1/a}, {(1,1)/b1/b}, {(2,1)/a2/a}, {(3,1)/b2/b}, {(4,1)/a3/a}, {(5,1)/a4/a}, {(6,1)/b3/b}, {(7,1)/a5/a}, {(8,1)/b4/b}, {(9,1)/a6/a}}
       { \node[vertex] (\name) at \pos {};
\node [above] at (\name.north) {$\label$};}

\foreach \source/ \dest  in {a2/b2,  a1/b1, b4/a5, a6/b4, a2/b1, b2/a3, a4/b3, b3/a5}
       \path[edge, black!50!white,  thick] (\source) --  (\dest);

\node at (2,1.75) {\large $D_j$};
\node at (7,1.75) {\large $D'_j$};

   \foreach \pos/\name / \label in {{(0.7,-2)/x1/x_1}, {(3.1,-2)/x2/x_{i_1}}, {(4.5,-2)/x3/x_{i_2}}, {(5.9,-2)/x4/x_{i_3}}, {(8.3,-2)/x5/x_n}}
       { \node[vertex] (\name) at \pos {};
\node [below] at (\name.south) {$\label$};}

\foreach \source/ \dest  in {b1/x1, b1/x2, b1/x3, b1/x4, b1/x5, b2/x1, b2/x2, b2/x3, b2/x4, b2/x5, b3/x1, b3/x2, b3/x3, b3/x4, b3/x5, b4/x1, b4/x2, b4/x3, b4/x4, b4/x5}
       \path[edge, black] (\source) --  (\dest);

\foreach \source/ \dest  in {x1/x2, x2/x3, x3/x4, x4/x5}
       \path[edge, black, line width = 1.6pt, dash pattern=on \pgflinewidth off 6pt] (\source) --  (\dest);

\foreach \source/ \dest  in {x2/a1, x2/a4, x3/a2, x3/a5, x4/a3, x4/a6}
       \path[edge, black, line width = 1.6pt] (\source) --  (\dest);

\end{tikzpicture}
\caption{An example of a graph $J_I$, as shown in~\cite{GPS12}. Only the clause $D_j=\{x_{i_1},x_{i_2},x_{i_3}\}$ is displayed.}\label{fig1}
\end{center}
\end{figure}

The following two lemmas are known.

\begin{lemma}[\cite{GPS12}]\label{l-truth}
The graph $J_I$ has a
colouring that respects $L$ if and only if $I$ has a satisfying truth
assignment in which each clause
contains at least one true and at least one false literal.
\end{lemma}

\begin{lemma}[\cite{GPS12}]\label{l-p6}
The graph $J_I$ is $P_6$-free.
\end{lemma}
 We are now ready to prove the main result of this section.

\begin{theorem}\label{t-main}
The {\sc List $4$-Colouring} problem is \NP-complete for the class of  $(C_5,C_6,K_4,\overline{P_1+2P_2},\overline{P_1+P_4},P_6)$-free graphs.
\end{theorem}

\begin{proof}
Lemma~\ref{l-truth} shows that the {\sc List 4-Colouring} problem is \NP-hard for the class of graphs $J_I$, where $I=(X,{\cal C})$ is an instance of  {\sc Not-All-Equal 3-Sat} with positive literals only, in which
every clause contains either two or three literals and in which each literal occurs in at most three different clauses.
Lemma~\ref{l-p6} shows that each $J_I$ is $P_6$-free.
As the {\sc List 4-Colouring} problem is readily seen to be in \NP, it remains to prove that  each $J_I$ contains no induced subgraph in the set $S=\{C_5,C_6,K_4,\overline{P_1+2P_2},\overline{P_1+P_4}\}$.
For contradiction, we assume that some $J_I$ has an induced subgraph~$H$ isomorphic to a graph in $S$.

First suppose that $H\in \{C_5,C_6\}$. The total number of $x$-type and $b$-type vertices can be at most $3$, as otherwise $H$ contains an induced $C_4$ or a vertex of degree at least~3.
Because $|V(H)|\geq 5$ and the subgraph of $H$ induced by its $b$-type and $x$-type vertices is connected, $H$ must contain at least two adjacent $a$-type vertices. This is not possible.

Now suppose that $H=K_4$. Because the $b$-type and $x$-type vertices induce a bipartite graph, $H$ must contain an $a$-type vertex.
Every $a$-type vertex has degree at most~3. If it has degree~3, then it has two non-adjacent neighbours (which are of $b$-type); a contradiction.

Finally suppose that $H\in \{\overline{P_1+2P_2},\overline{P_1+P_4}\}$. Let $u$ be the vertex that has degree~4 in~$H$. Then $u$ cannot be of $a$-type, because no $a$-type vertex has more than three neighbours in $J_I$.
If $u$ is of $b$-type, then every other vertex of~$H$ is either of $a$-type or of $x$-type, and because vertices of the same type are not adjacent, $H$ must contain two $a$-type vertices and two $x$-type vertices. Then an $a$-type vertex is adjacent to two
$x$-type vertices; a contradiction.
Thus $u$ must be of $x$-type, and so every other vertex of $H$ is either of $a$-type or of $b$-type. Because vertices of the same type are non-adjacent, $H$ must contain two $a$-type vertices and two $b$-type vertices. But then $u$ is adjacent to two $a$-type vertices in the same
 clause-component. This is not possible.
\qed
\end{proof}

Our second result modifies the same reduction.
From $J_I$ and $L$, we obtain a new graph $J_I'$ and list assignment $L'$ by subdividing every edge between an $a$-type vertex and an $x$-type vertex  and giving each new vertex  the list $\{1,2\}$.
We say that these new vertices are of $c$-type.

\begin{lemma}\label{l-bipartitej}
The graph $J_I'$ is $P_8$-free and chordal bipartite.
\end{lemma}

\begin{proof}
We first prove that $J_I'$ is $P_8$-free (but not $P_7$-free).
Let $P$ be an induced path in $J_I'$.
If $P$ contains no $x$-type vertex, then $P$ contains vertices of at most one clause-component together with at most two $c$-type vertices. This means that $|V(P)|\leq 7$.
If $P$ contains no $b$-type vertex, then~$P$ can contain at most one $x$-type vertex (as any two $x$-type vertices can only be connected by a path that uses at least one $b$-type vertex). Consequently,~$P$ can have at most two $a$-type vertices and at most two $c$-type vertices. Hence, $|V(P)|\leq 5$ in this case.
From now on assume that~$P$ contains at least one $b$-type vertex and at least one $x$-type vertex. Also note that~$P$ can contain in total at most three vertices of $b$-type and $x$-type.

First suppose that $P$ contains exactly three vertices of $b$-type and $x$-type. Then these vertices form a 3-vertex subpath in $P$ of types $b,x,b$ or $x,b,x$. In both cases we can extend both ends of the subpath only by an $a$-type vertex and an adjacent $c$-type vertex, which means that $|V(P)|\leq 7$.
Now suppose that~$P$ contains exactly two vertices of $b$-type and $x$-type. Because these vertices are of different type, they are adjacent and we can extend both ends of the corresponding 2-vertex subpath of $P$ only by an $a$-type vertex and an adjacent $c$-type vertex.  This means that $|V(P)|\leq 6$.
We conclude that $J_I'$ is $P_8$-free.

We now prove that $J_I'$ is chordal bipartite.
The graph $J_I'$ is readily seen to be bipartite.
Because $J_I'$ is $P_8$-free, it contains no induced cycle on ten or more vertices.
Hence, in order to prove that $J_I'$ is chordal bipartite, it remains to show that $J_I'$ is $(C_6,C_8)$-free.

For contradiction, let $H$ be an induced subgraph of $J_I'$ that is isomorphic to $C_6$ or $C_8$.
Then $H$ must contain at least one vertex of $x$-type and at least one vertex of $b$-type.
The subgraph of $H$ induced by its vertices of $x$-type and $b$-type is connected and has size at most~3. This means
that the subgraph of $H$ induced by its vertices of $a$-type and $c$-type is also connected and has size at least~3.
This is not possible.
We conclude that $J_I'$ is chordal bipartite.\qed
\end{proof}

The following lemma can be proven by exactly the same arguments that were used to prove Lemma~\ref{l-truth}.

\begin{lemma}\label{l-truthbip}
The graph $J_I'$ has a
colouring that respects $L'$ if and only if $I$ has a satisfying truth
assignment in which each clause
contains at least one true and at least one false literal.
\end{lemma}

Lemmas~\ref{l-bipartitej} and~\ref{l-truthbip} imply the last result of this section.

\begin{theorem}\label{t-hard4}
{\sc List $4$-Colouring} is \NP-complete for $P_8$-free chordal bipartite graphs.
\end{theorem}

\section{New Results for Precolouring Extension}\label{s-precol}

In this section we give three results on $k$-{\sc Precolouring Extension}.

Let $k\geq 4$.
Consider the bipartite graph $J_I'$ with its list assignment  $L'$ from Section~\ref{s-4col} that was defined immediately before Lemma~\ref{l-bipartitej}.
The list of admissible colours $L'(u)$ of each vertex $u$ is a subset of $\{1,2,3,4\}$. We add~$k-|L'(u)|$ pendant vertices to $u$ and precolour these vertices with different colours from $\{1,\ldots,k\}\setminus L'(u)$.
This results in a graph $J_I^k$ with
a $k$-precolouring $c_{W^k}$,
where~$W^k$ is the set of all the new pendant vertices in $J_I^k$.

\begin{lemma}\label{l-bipartitej2}
For all $k\geq 4$, the graph $J_I^k$ is $P_{10}$-free and
chordal bipartite.
\end{lemma}

\begin{proof}
Because $J_I'$ is $P_8$-free and chordal bipartite by Lemma~\ref{l-bipartitej}, and we only added pendant vertices,
$J_I^k$ is $P_{10}$-free and chordal bipartite.\qed
\end{proof}

The following lemma can be proven by the same arguments that were used to prove Lemmas~\ref{l-truth} and~\ref{l-truthbip}.

\begin{lemma}\label{l-truthbip2}
For all $k\geq 4$, the graph $J_I^k$ has a $k$-colouring that is an extension of $c_{W^k}$
if and only if $I$ has a satisfying truth
assignment in which each clause
contains at least one true and at least one false literal.
\end{lemma}
Lemmas~\ref{l-bipartitej2} and~\ref{l-truthbip2} imply the first result of this section which extends a result of Kratochv\'{\i}l~\cite{Kr93} who showed that
5-{\sc Precolouring Extension} is \NP-complete for $P_{13}$-free bipartite graphs.

\begin{theorem}\label{t-prebipartite}
For all $k\geq 4$,  $k$-{\sc Precolouring Extension}  is \NP-complete for the class of $P_{10}$-free chordal bipartite graphs.
\end{theorem}
Here is our second result.

\begin{theorem}\label{t-new1}
The {\sc $4$-Precolouring Extension} problem is \NP-complete for the class of $(C_5,C_6,C_7,C_8,P_8)$-free graphs.
\end{theorem}

\begin{proof}
Let $J_I$ be the instance with list assignment $L$ as constructed at the start of Section~\ref{s-4col}. Instead of considering lists, we introduce new vertices, which we precolour (we do not precolour any of the original vertices).
For each clause component $D_j$ we add five new vertices, $s_j$, $t_j$, $u_{j,1}$, $u_{j,2}$, $u_{j,3}$. We add edges $a_{j,1}s_j$, $a_{j,3}t_j$ and $a_{j,h}u_{j,h}$ for $h=1,\ldots 3$. We precolour $s_j$, $t_j$, $u_{j,1}$, $u_{j,2}$, $u_{j,3}$ with colours $3$, $4$, $1$, $1$, $1$, respectively.
For each clause component $D_j'$ we add five new vertices, $s_j'$, $t_j'$, $u_{j,1}'$, $u_{j,2}'$, $u_{j,3}'$. We add edges $a_{j,1}'s_j'$, $a_{j,3}'t_j'$ and $a_{j,h}'u_{j,h}'$ for $h=1,\ldots 3$. We precolour $s_j$, $t_j$, $u_{j,1}$, $u_{j,2}$, $u_{j,3}$ with colours $3$, $4$, $2$, $2$, $2$, respectively.
Finally, we add two new vertices $c_1,c_2$, which we make adjacent to all $x$-type vertices, and two new vertices $y_1,y_2$, which we make adjacent to all
$b$-type vertices. We colour $c_1$, $c_2$, $y_1$, $y_2$ with colours~$3$, $4$, $1$, $2$, respectively.
This results in a new graph $J_I^*$.
Because $y_1,y_2$ can be viewed as $x$-type vertices and $c_1,c_2$ as $b$-type vertices, because every other new vertex is a pendant vertex and because $J_I$ is $(C_5,C_6,P_6)$-free (by Theorem~\ref{t-main}),
we find that $J_I^*$ is $(C_5,C_6,C_7,C_8,P_8)$-free.
Moreover, our precolouring forces the list $L(v)$ upon every vertex $v$ of $J_I$.
Hence, $J_I^*$ has a 4-colouring extending this precolouring if and only if $J_I$ has a colouring that respects $L$.
By Lemma~\ref{l-truth}  the latter is true if and only if $I$ has a satisfying truth
assignment in which each clause contains at least one true and at least one false literal.\qed
\end{proof}

Broersma et al.~\cite{BFGP13} showed that  $5$-{\sc Precolouring Extension} for $P_6$-free graphs is \NP-complete. It can be shown that the gadget constructed in their
\NP-hardness reduction is $(C_5,C_6)$-free.
By adding $k-5$ dominating vertices, precoloured with colours $6,\ldots,k$, to each vertex in their gadget, we can extend their result from $k=5$ to $k\geq 5$.
This leads to the following theorem.

\begin{theorem}\label{t-new2}
For all $k\geq 5$, {\sc $k$-Precolouring Extension} is \NP-complete for $(C_5,C_6,P_6)$-free graphs.
\end{theorem}

\section{New Results for Colouring} \label{s-colouring}

The first result that we show in this section is that $4$-{\sc Colouring}
is \NP-complete for  $(C_3,P_{22})$-free graphs.
We will prove this by modifying the graph $J_I^4$ from Section~\ref{s-precol}.   It improves
the result of Golovach~et al.~\cite{GPS11} who showed that 4-{\sc Colouring} is \NP-complete for $(C_3,P_{164})$-free graphs.

First we review a well-known piece of graph theory.  The {\em Mycielski construction} of a graph $G$ is created by adding, for each vertex $v$ of $G$, a new vertex that is adjacent to each vertex $N_G(v)$, and then adding a further vertex that is adjacent to each of the other new vertices.  For example, the Mycielski construction of $K_2$ is a 5-cycle, and the Mycielski construction of a 5-cycle is the well-known Gr\"{o}tzsch graph.  These examples are the first in an infinite sequence of graphs $M_2,M_3,\ldots$ where $M_2=K_2$ and $M_i$, $i \geq 3$, is the Mycielski construction of $M_{i-1}$.   The graph~$M_5$ is displayed in Figure~\ref{fig2}.
As we will make considerable use of this graph, let us explain its construction carefully.  Let us suppose that we start with $M_3$ where $V(M_3)=\{1,2,3,4,5\}$ and $E(M_3)=\{\{1,2\},\{2,3\},\{3,5\},\{5,4\},\{4,1\}\}$
(note that, for clarity, we denote an edge between two vertices $u$ and $v$ by $\{u,v\}$ instead of $uv$).
Then $M_4$ is made by adding each vertex $i$, $6 \leq i \leq 10$, and making it adjacent to the neighbours of vertex $i-5$ and to a further vertex~11.  Finally~$M_5$ is obtained by adding a vertex $i$, $12 \leq i \leq 22$, with the same neighbours as vertex $i-11$ and a further last new vertex, 23.  Mycielski~\cite{My55} showed that each $M_k$ is $C_3$-free and has chromatic number $k$.  Moreover, any proper subgraph of~$M_k$ is $(k-1)$-colourable (see for example \cite{BM}).

\tikzstyle{vertex}=[circle,draw=black, fill=black, minimum size=5pt, inner sep=1pt]
\tikzstyle{edge} =[draw,-,black,>=triangle 90]

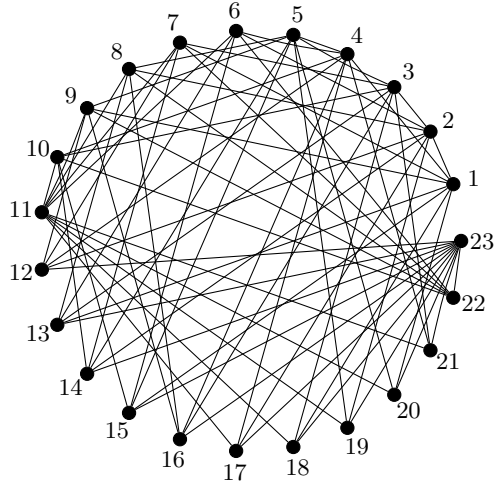
\begin{figure}
\begin{center}
\begin{tikzpicture}[xscale=0.7,yscale=0.7]

\foreach \s  in {1, ..., 23}
{  \node[vertex, fill=black] (\s) at (360/23 * \s:4) {};
\node at (360/23 * \s:4.4 ) {\s};
}

\foreach \source/ \dest  in {5/15,1/2, 3/2, 4/1, 5/4,  6/2, 7/1, 7/3, 8/2, 8/5, 9/1, 9/5, 10/4, 10/3, 11/6, 11/7, 11/8, 11/10,
12/2, 3/13, 4/12, 5/14, 6/13, 6/15, 7/12, 7/14, 8/13, 8/16, 9/12, 9/16, 10/15, 10/14, 11/17, 11/18, 11/19, 11/20, 11/21,
1/13, 14/2, 15/1, 16/4, 16/3, 17/2, 17/4, 18/1, 18/3, 19/2, 19/5, 20/1, 20/5, 21/4, 21/3, 22/6, 22/7, 22/8, 22/9, 22/10,
23/12, 23/13, 23/14, 23/15, 23/16, 23/17, 23/18, 23/19, 23/20, 23/22}
       \path[edge, black] (\source) --  (\dest);

\draw[edge, black] (3) .. controls (360/23 * 4:3.8 ) .. (5);
\draw[edge, black] (21) .. controls (360/23 * 22:3.8 ) .. (23);
\draw[edge, black] (4) .. controls (360/23 * 5:3.8 ) .. (6);
\draw[edge, black] (9) .. controls (360/23 * 10:3.8 ) .. (11);

\end{tikzpicture}
\caption{The graph $M_5$: vertex sets $\{1,2\}$, $\{1,2,3,4,5\}$ and $\{1,2,3,4,5,6,7,8,9,10,11\}$ induce $M_2$, $M_3$ and $M_4$,
respectively.}\label{fig2}
\end{center}
\end{figure}

Let $M'$ be the graph obtained from $M_5$ by removing the edge from 17 to 23.  We need the following lemma.
\begin{lemma}
Let $\phi$ be a $4$-colouring of $M'$.  Then $\phi(17) = \phi(23)$, and moreover,  $\{\phi(2), \phi(4), \phi(11), \phi(17)\} = \{\phi(2), \phi(4), \phi(11), \phi(23)\} = \{1, 2, 3, 4\}$.
\end{lemma}

\begin{proof}
If $\phi(17)\neq \phi(23)$, then $\phi$ is a 4-colouring of $M_5$, which is not possible.
 As $N_{M'}(17) = \{2, 4, 11\}$, we must have $\{\phi(2), \phi(4), \phi(11), \phi(17)\} = \{1, 2, 3, 4\}$ else there is a 4-colouring $\phi'$ which disagrees with $\phi$ only on 17 and so, again, is a 4-colouring of $M_5$.  Thus we also have $\{\phi(2), \phi(4), \phi(11), \phi(23)\} = \{1, 2, 3, 4\}$.\qed
\end{proof}

Let $T=\{t_1,t_2,t_3,t_4\}=\{2,4,11,23\}$.  We present three lemmas about induced paths in $M'$ with endvertices in $T$.
Proving each of these three lemmas is straightforward but tedious. We therefore omitted their 
proofs.

\begin{lemma}  \label{l-mpropalt1a}
Every induced path in $M'$ with both endvertices in $T$ contains at most $7$ vertices.
\end{lemma}

\begin{lemma}  \label{l-mpropalt1b}
Every induced path in $M'$ with an endvertex in $T$ contains at most~$8$ vertices.
\end{lemma}

\begin{lemma}  \label{l-mpropalt2}
$M'$ contains no induced subgraph isomorphic to $P_8+P_1$ such that each of the two paths has an endvertex in $T$.
Also, $M'$ contains no induced subgraph isomorphic to $2P_7$ such that each of the two paths has an endvertex~in~$T$.
\end{lemma}

For a graph $G=(V,E)$ and subset $U\subseteq V$, we let $G-U$ denote the graph obtained from $G$ after removing all vertices in $U$.

To prove our result we also use the graphs $J'_I$ and $J_I^4$ defined in Sections~\ref{s-4col} and~\ref{s-precol}, respectively.
Let $A$, $B$, $C$, $X$ denote the sets of of $a$-type, $b$-type, $c$-type and $x$-type vertices in $J'_I$, respectively. Note these sets also exist in $J_I^4$.
We need an upper bound on the length of an induced path in some specific induced subgraphs of $J'_I$ and $J^4_I$.

\begin{lemma} \label{subpathJI}
Every induced path in $J'_I-C$ that starts with an $a$-type or $x$-type vertex has at most~$5$ vertices.
\end{lemma}

\begin{proof}
Let $P=v_1v_2,\ldots v_r$ be a maximal induced path in $J'_I$.
Suppose that~$v_1$ is of $a$-type.
Then, as $P$ contains no $c$-type vertices, we find that $v_2=b$. Consequently, $v_3=a$ or $v_3=x$.  In both cases we must have $v_4=b$ and $v_5=a$. Afterward we cannot extend $P$
any further.
Suppose that $v_1$ is of $x$-type.
Because there is no $c$-type vertex in $P$, we find that $v_2=b$. Hence, $v_3=a$ or $v_3=x$. In both cases, $P$ cannot be extended further.
We conclude that $r\leq 5$.
\qed
\end{proof}

\begin{lemma} \label{subpathJ4}
Every induced path in $J^4_I-B$ has\\[-18pt]
\begin{enumerate}
\item [(i)] at most $7$ vertices;
\item [(ii)] at most $6$ vertices if it contains only one pendant vertex of $J^4_I-B$;
\item [(iii)] at most $5$ vertices if it contains no pendant vertex of $J^4_I-B$.
\end{enumerate}
\end{lemma}

\begin{proof}
We observe that $J^4_I-B$ is a forest in which each tree can be constructed as follows.
Take a star. Subdivide each of its edges exactly once. Afterward add one or more pendant vertices to each vertex.
The tree obtained is $P_8$-free and every induced $P_7$ contains two pendant vertices as its end-vertices.
Moreover, every induced $P_6$ has at least one pendant vertex. \qed
\end{proof}

We will now modify the (bipartite) graph $J_I^4$ with the precolouring $c_{W^4}$ in the following way.
\begin{itemize}
\item [$\bullet$] Remove all vertices that are pendant to any vertex in $B\cup C$.
\item [$\bullet$] Add a copy of $M'$. Write $t_1=2$, $t_2=4$, $t_3=11$ and $t_4=23$.
\item [$\bullet$] Let $S$ be the set of vertices pendant to any vertex in $A\cup X$.
For each $v\in S$ do as follows. If $c_{W^4}(v)=i$ then add the edge $vt_j$ for all $j\in \{1,2,3,4\}\setminus \{i\}$.
\item [$\bullet$] Add an edge between every vertex in~$B$ and $t_i$ for $i=1,2$.
\item [$\bullet$] Add an edge between every vertex in~$C$ and $t_i$ for $i=3,4$.
\end{itemize}
We call the resulting graph $J^*_I$. 
Note that $J^*_I$ is not $P_{21}$-free in general. In order to see this we say that the pendant vertices of $B\cup C$ are of type~$p$ and that 
the vertices of $T$ are of type $t$ and then take an induced 21-vertex path
with vertices of type $$a-c-x-c-a-p-t-p-a-c-x-c-a-p-t-p-a-c-x-c-a.$$
Note that such a path only uses two vertices of $M'$, namely $t_1$ and $t_2$ (as $t_3$ and $t_4$ are adjacent to all vertices of $c$-type).
As such, trying to optimize Lemmas~\ref{l-mpropalt1a}--\ref{l-mpropalt2} (which we believe is possible) does not help us with improving our result.
In the following lemma we show that this length is maximum.

\begin{lemma}\label{l-star1}
The graph $J^*_I$ is $(C_3,P_{22})$-free.
\end{lemma}

\begin{proof}
Because $S\cup B\cup C$ and $\{t_1,t_2,t_3,t_4\}$ are both independent sets and the graphs $M'$ and~$J^4_I$ are
bipartite,  $J^{*}_I$ is $C_3$-free.
Below we show that $J^{*}_I$ is $P_{22}$-free.
Let $P$ be an induced path in $J^{*}_I$. Let $\alpha$ be the number of vertices of $\{t_1,t_2,t_3,t_4\}$ in $P$.

\medskip
\noindent {\bf Case 1.} $\alpha=0$.\\
Then either $P\subseteq J^4_I$, and so $|V(P)|\le 9$ by Lemma~\ref{l-bipartitej2}, or $P\subseteq M'$, and so $|V(P)|\le 21$.

\medskip
\noindent {\bf Case 2.}  $\alpha=1$.\\
Then $P=P_Lt_iP_R$ for some $1\leq i\leq 4$, where the subpaths $P_L$ and $P_R$ are each fully contained in either $M'$ or $J^4_I$. If both of them are contained in $M'$, then $P\subseteq M'$ and so $|V(P)|\le 21$.
If one of them is contained in $M'$ and the other one is contained in $J^4_I$,
then $|V(P)|\le 8+9=17$ by Lemmas~\ref{l-bipartitej2} and~\ref{l-mpropalt1b}. Otherwise both of them are
contained in $J^4_I$ and so $|V(P)|\le 9+1+9=19$ by Lemma~\ref{l-bipartitej2}.

\medskip
\noindent {\bf Case 3.}  $\alpha=2$.\\
Then $P=P_Lt_iP_{M}t_jP_R$ for some $i,j\in \{1,2,3,4\}$, where each of $P_L$, $P_M$ and~$P_R$ is fully contained
in either $M'$ or $J^4_I$. We need the following claim.

\medskip
\noindent {\em Claim 1.
$|V(P_L)|\le 6$ and $|V(P_R)|\le 6$.}

\medskip
\noindent
We prove Claim~1 as follows.
By symmetry it suffices to show it for $P_L$ only.
First assume that $P_L$ is contained in $J^4_I$.
We distinguish two cases.

\medskip
\noindent {\bf Case a.} $t_i=t_1$.\\
Let $t^{-}_1$ be the right endvertex of $P_L$.
First assume that $t^{-}_1\notin B$.
Then $t^{-}_1$ must be in~$S$ by the construction of $J^{*}_I$.
Because $t_1$ is adjacent to every vertex in~$B$, we have $V(P_L)\cap B=\emptyset$.
Hence, $P_L\subseteq J^4_I-B$.
We also  have $|V(P_L)\cap S|=1$, as~$P$ is induced and each vertex in $S$ is adjacent to three vertices of $T$.
Then, by Lemma~\ref{subpathJ4}, we obtain $|V(P_L)|\le 6$.
Now assume $t^{-}_1\in B$.
Then $|V(P_L)\cap S|=0$ for the same reason as before.
Hence $P_L-\{t^-_1\}$ contains only non-pendant vertices of $J^4_I-B$.
Then, by Lemma \ref{subpathJ4}, we obtain $|V(P_L)|\le 5+1=6$.

\medskip
\noindent {\bf Case b.} $t_i=t_3$.\\
Let $t^{-}_3$ be the right endvertex of $P_L$.
If $t_j\in \{t_1,t_2\}$ then $P_L\cap B=\emptyset$, and so~$P_L$ is an induced path of $J^4_I-B$.
By Lemma~\ref{subpathJ4} we find that $|V(P_L)|\le 6$.
Now assume that $t_j=t_4$. Since $t_4$ is adjacent to all vertices of~$C$, we have $V(P_L)\cap C=\emptyset$.
Hence, by construction, we find that $t^-_3\in S$. In fact we have $V(P)\cap S=\{t^{-}_3\}$,
as~$P$ is induced and each vertex in $S$ is adjacent to three vertices of $T$.
Thus, $P_L-\{t^{-}_3\}$ is contained in $J'_I-C$.
As $t^-_3\in S$, its neighbour on $P_L$ must be in $A\cup X$ (if this neighbour exists).
Then, by Lemma~\ref{subpathJI}, we find that $|V(P_L)|\le 5+1=6$.

\medskip
\noindent
Now suppose that $P_L$ is contained in~$M'$.
Using Lemma \ref{l-mpropalt2} with respect to the paths
$P_Lt_i$ and $t_j$ we conclude
that $|V(P_L)|\le 6$. This completes the proof of Claim~1.

\medskip
\noindent
By Claim~1 we have that  $|V(P_L)|\le 6$ and $|V(P_R)|\le 6$.
If $|V(P_M)|\le 7$ this means that $|V(P)|\le 6+1+7+1+6=21$ and we are done.
So, it suffices to prove that $|V(P_M)|\le 7$, as we will do below.

If $P_M$ is contained in $M'$, then by Lemma \ref{l-mpropalt1a}
we have that $|V(P_M)|\le 5$. Assume that $P_M\subseteq J^4_I$.
First suppose that $\{t_i,t_j\}=\{t_1,t_2\}$. If $P_M$ contains a $b$-type vertex,
then $P_M$ must be a single $b$-type vertex and thus $|V(P_M)|=1$.
Otherwise we have $P_M\subseteq J^4_I-B$, and hence $|V(P_M)|\le 7$ by Lemma \ref{subpathJ4}.

Now suppose that $\{t_i,t_j\}=\{t_3,t_4\}$. If $P_M$ contains a $c$-type vertex,
then~$P_M$ must be a single $c$-type vertex and thus $|V(P_M)|=1$.
Otherwise
both end-vertices of $P_M$ are in $S$
and all internal vertices of $P_M$ are contained in $J'_I-C$.
Because every neighbour of a vertex in $S$ in $J'_I$ belongs to $A\cup X$,
we find that $|V(P_M)|\le 1+5+1=7$ by Lemma~\ref{subpathJI}.

We observe that the pairs $\{t_1,t_2\}$ and $\{t_3,t_4\}$
are symmetric in terms of the edges between them to $B$ and $C$.
Moreover, $t_1$ and $t_2$ can be seen as symmetric, and the same holds for $t_3$ and $t_4$.
This has the following consequence.
Suppose that  $\{t_i,t_j\}\notin \{\{t_1,t_2\},\{t_3,t_4\}\}$. Then we may assume without loss of generality that
 $\{t_i,t_j\}=\{t_1,t_3\}$, say $P=P_Lt_1P_Mt_3P_R$.

If the left endvertex of $P_M$ is not in~$B$,
then $P_M\subseteq J^4_I-B$ and so $|V(P_M)|\le 7$ by Lemma \ref{subpathJ4}.
Similarly, if the right endvertex of $P_M$ is not in~$C$ then all vertices except the left endvertex of $P_M$ are in $J'_I-C$.
In particular, the right endvertex of $P_M$ belongs to $S$ in that case, and as such has its neighbour on $P_M$ in $A\cup X$ (if this neighbour exists).
Then $|V(P_M)|\le 1+5=6$ by Lemma~\ref{subpathJI}.
Therefore we may assume that $P_M$ starts with a vertex in $B$, ends at a vertex in $C$.
As $P$ is induced and each vertex in $S$ is adjacent to three vertices of $T$, we find that $P_M$ contains no vertices of~$S$.
Moreover, no internal vertex of $P_M$ belongs to $B\cup C$, as these vertices would be adjacent to one of $\{t_1,t_3\}$.
Hence all internal vertices of $P_M$ are in $A\cup X$.
Then $P_M$ only has one internal vertex, which is either in $A$ or in $X$. Hence, $|V(P_M)|=3\leq 7$. This completes the proof of Case~3.

\medskip
\noindent {\bf Case 4.}  $\alpha=3$.\\
Then $P=P_Lt_hP^{1}_Mt_iP^{2}_Mt_jP_R$ for some $h,i,j\in \{1,2,3,4\}$, where the subpaths $P_L$, $P^1_M$, $P^2_M$ and $P_R$ are each fully contained in either $M'$ or $J^4_I$.
As $P$ is induced and each vertex in $S$ is adjacent to three vertices of $T$.
we find that $|V(P^{i}_M)|=1$ if $P^{i}_M$ ($i=1,2$) contains a vertex of~$S$.
We observe that Claim~1 is also valid here (as exactly the same arguments can be used).
Hence, $|V(P_L)|\le 6$ and $|V(P_R)|\le 6$.
By the aforementioned symmetry relations between vertices and vertex pairs of $\{t_1,t_2,t_3,t_4\}$, we may assume without loss of generality that $t_1,t_3,t_4\in P$. According to the relative positions
among these three vertices on the path $P$ we have the following two subcases.

\medskip
\noindent
{\bf Case 4.1} $P=P_Lt_1P^{1}_Mt_3P^{2}_Mt_4P_R$.\\
We first prove that either $P^{1}_M$ is contained in $M'$ or that $|V(P^{1}_M)|=1$. Suppose that $P^1_M$ is not contained in $M'$.
The right endvertex of $P^1_M$ cannot be in $B$ due to the presence of $t_4$. Hence, it must be in $S$. As observed above, in that case we have $|V(P^{i}_M)|=1$.
We now prove that either $P^{2}_M$ is contained in $M'$ or that $|V(P^{2}_M)|=1$.
Suppose that~$P^{2}_M$ is not contained in $M'$. Then $P^{2}_M$ is contained in $J^4_I$. As mentioned above,
$|V(P^{2}_M)|=1$ if $P^{2}_M$ contains a vertex of~$S$. So, assume $P^{2}_M$ contains no vertex of~$S$.
Then $P^{2}_M$ must contain a vertex of $C$ and consequently cannot contain any other vertex. Hence $|V(P^{2}_M)|=1$.

Now, if each $P^{i}_M$ is contained in $M'$ then $|V(P)|\le 6+7+6=19$ by  Lemma~\ref{l-mpropalt1a}
and the aforementioned fact that  $|V(P_L)|\le 6$ and $|V(P_R)|\le 6$.
Suppose that $P^1_M$ is not contained in $M'$. Then $|V(P^{1}_M)|=1$. If  $P^{2}_M$ is contained in $M'$
then $t3P^2_Mt_4$ has at most seven vertices by Lemma~\ref{l-mpropalt2}. Thus, $|V(P)|\le 6+1+1+7+6=21$. Otherwise, $|V(P^2_M)|=1$, and consequently, $|V(P)|\leq 17$.
If $P^2_M$ is not contained in $M'$, we can repeat the arguments.

\medskip
\noindent
{\bf Case 4.2}  $P=P_Lt_3P^{1}_Mt_1P^{2}_Mt_4P_R$.\\
Note that $P$ contains no vertices of $B$ and recall that $|V(P^{i}_M)|=1$ if $P^{i}_M$ ($i=1,2$) contains a vertex of~$S$. Hence, $P$ contains no vertices of $C$ either,
and the claim that either $P^{i}_M$ is contained in $M'$ or $|V(P^{i}_M)|=1$ is still true for $i=1,2$.
By repeating the arguments as in Case 4.1 we find that $|V(P)|\le 21$.

\medskip
\noindent {\bf Case 5.}  $\alpha=4$.\\
Then $P=P_Lt_hP^{1}_Mt_iP^{2}_Mt_jP^{3}_Mt_kP_R$ for $\{h,i,j,k\}=\{1,2,3,4\}$, where the subpaths $P_L$, $P^1_M$, $P^2_M$, $P^3_M$ and $P_R$ are each fully contained in either $M'$ or $J^4_I$.  Note that $V(P)\cap S=\emptyset$ as $P$ is induced and each vertex in $S$ is adjacent to three vertices of $T$.

By symmetry, it suffices to consider the following three cases.

\medskip
\noindent
{\bf Case 5.1.} $P=P_Lt_1P^{1}_Mt_3P^{2}_Mt_2P^{3}_Mt_4P_R$.\\
Then $P$ cannot contain any vertex of $B\cup C$. Consequently, as $P$ contains no vertex from $S$ either,
$P$ is a path in $M'$. Hence, $|V(P)|\le 21$.

\medskip
\noindent
{\bf Case 5.2.} $P=P_Lt_1P^{1}_Mt_3P^{2}_Mt_4P^{3}_Mt_2P_R$.\\
Then $P$ contains no vertex of $C$, and only $P^2_M$ may contain a vertex of $B\cup C$, which must be in~$B$.
Hence, as $P$ contains no vertex of $S$, all other three subpaths are fully contained in $M'$.
If $P^2_M$ contains no vertex of $B$, then $P$ is a path in $M'$ and hence $|V(P)|\leq 21$.
Otherwise, $|V(P^{2}_M)|=1$. Then $|V(P)|\le 8+1+8=17$ by Lemma \ref{l-mpropalt1b}.

\medskip
\noindent
{\bf Case 5.3.} $P=P_Lt_1P^{1}_Mt_2P^{2}_Mt_3P^{3}_Mt_4P_R$.\\
As $P$ has no vertex of $S$, all paths $P_L$, $P^{2}_M$, $P_R$ must belong to $M'$, whereas
$P^{1}_M$ and $P^3_M$ are either a single vertex (not in $M'$) or contained in $M'$.
If both of $P^{1}_M$ and $P^{3}_M$ are contained in $M'$, then $P$ is contained in $M'$, and
thus $|V(P)|\le 21$. If exactly one of $P^{1}_M$ and $P^{3}_M$ is contained in $M'$,
then  $|V(P)|\le 8+1+8=17$ by Lemma \ref{l-mpropalt1b}. Thus $|V(P^{1}_M)|=1$ and $|V(P^{3}_M)|=1$.
By Lemma~\ref{l-mpropalt2} we find that $P_L$, $P^2_M$ and $P_L$ each have at most seven vertices.
Moreover, Lemma~\ref{l-mpropalt2} also implies that the sum of the order of any two of these subpaths is at most 13. Hence, $|V(P)|\le 3\times 13/2+2=21.5$.
Thus, $|V(P)|\le 21$. This completes the proof of Case 5, and as such we have proven the lemma.\qed
\end{proof}
The following lemma follows from the way we constructed $J^*_I$ from $J^4_I$.
\begin{lemma}\label{l-star2}
The graph $J^*_I$ has a $4$-colouring if and only if $J^4_I$ has a $4$-colouring that is an extension of $c_{W^k}$.
\end{lemma}
Combining Lemmas~\ref{l-truthbip2},~\ref{l-star1} and~\ref{l-star2} gives us the main result of this section.

\begin{theorem}\label{t-new3}
$4$-{\sc Colouring} is \NP-complete for $(C_3,P_{22})$-free graphs.
\end{theorem}

To prove the final result of this section we need the following classical result of Erd\"os~\cite{Er59} as a lemma (we follow the same approach
as that used by Kr\'al' et al.~\cite{KKTW01} and Kami\'nski and Lozin~\cite{KL07} for proving that for all $k\geq 3$ and $s\geq 3$,
 {\sc $k$-Colouring}  is \NP-complete for
$(C_3,\ldots,C_s)$-free graphs).

\begin{lemma}[\cite{Er59}]\label{l-erdos}
For every pair of integers $g$ and $k$, there exists
a graph with girth~$g$ and chromatic number~$k$.
\end{lemma}

This final result extends a result of Golovach et al.~\cite{GPS11} who proved that for all $s\geq 6$, there exists a constant $t^s$ such that $4$-{\sc Colouring} is
 \NP-complete for $(C_5,\ldots,C_{s-1},P_{t^s})$-free graphs.
It is also known that, for all $k\geq 3$ and all $s\geq 4$, {\sc $k$-Colouring}  is \NP-complete for graphs of girth at least~$s$, or equivalently,
$(C_3,\ldots,C_{s-1})$-free 
graphs.  This has been shown by Kr\'al', Kratochv\'{\i}l, Tuza, and
Woeginger~\cite{KKTW01} for the case $k=3$ and by Kami\'nski and Lozin~\cite{KL07} for the case $k\geq 4$.
On the other hand, Golovach et al.~\cite{GPS11} showed that for all $k\geq 1$ and $r\geq 1$,
even {\sc List $k$-Colouring} is p-time solvable for $(C_4,P_r)$-free graphs (also see Theorem~\ref{t-coloringall}), and thus for 
$(C_3,C_4,\ldots,C_{s-1},P_r)$-free graphs for all $s\geq 5$. As such, our new hardness result is best possible and
can be seen as an analog (for $k\geq 4$) of the aforementioned result shown by
Kr\'al' et al.~\cite{KKTW01} and Kami\'nski and Lozin~\cite{KL07}.

\begin{theorem}\label{t-final}
For all $k\geq 4$ and $s\geq 6$, there exists a constant $t^s_k$ such that $k$-{\sc Colouring} is \NP-complete for
$(C_3,C_5,\ldots,C_{s-1},P_{t^s_k})$-free graphs.
\end{theorem}

\begin{proof}
Let $k\geq 4$ and $s\geq 6$.
By Lemma~\ref{l-erdos}, there exists an edge-minimal graph~$F$ with chromatic number $k+1$ and girth $s$.
Let $pq$ be an edge in $F$ and let $F-pq$ denote the graph obtained from $F$ after removing $pq$.
Then, by definition, $F-pq$ has at least one $k$-colouring, and moreover, $p$ and $q$ receive the same colour in every $k$-colouring of $F-pq$.
We introduce a new vertex $q^*$ that we make adjacent only to $q$. Call the resulting graph $F'$.
Then $F'$ has at least one $k$-colouring, and moreover, $p$ and $q^*$ receive a different colour in every $k$-colouring of $F'$.
An {\em $F'$-identification} of some edge $uv$ in a graph $G$ is the following operation: delete the edge $uv$ and add a copy of $F'$ between $u$ and $v$ by identifying vertices $u$ and $v$ with $p$ and $q^*$, respectively (we call these two new vertices $u$ and $v$ again).

Now consider the graph $J_I^4$.
We take a complete graph on $k$ new vertices $r_1,\ldots,r_k$.
Recall that we had defined a precolouring~$c_{W^4}$ for the subset $W^4\subseteq V(J_I^4)$.
We add an edge between a vertex $r_i$ and a vertex $u\in W^4$ if and only if $c_{W^4}(u)\neq i$.
Afterward we perform an $F'$-identification of every edge between two vertices $r_i$ and $r_j$ and on every edge between
a vertex $r_i$ and a vertex in~$W^4$. Let $G_I^4$ be the resulting graph.

We observe that $G_I^4$ is not $C_4$-free.
However, because $J^4_I$ is chordal bipartite by Lemma~\ref{l-bipartitej2} and because we performed appropriate $F'$-identifications,
$G_I^4$ is  $(C_3,C_5,\ldots,C_{s-1})$-free.
Let $Q$ be an induced path in $G_I^4$.
Let $h=|V(Q)|\cap \{r_1,\ldots,r_k\}$.
Then we can write $Q=Q_1r_{i_1}Q_2r_{i_2}\cdots Q_hr_{i_h}Q_{h+1}$, where the vertices of each $Q_i$ all belong either to an $F'$-copy or to $J^4_k$.
Because $|F'|$ is a constant that depends only on $k$ and~$s$, and $J^4_I$ is $P_{10}$-free by Lemma~\ref{l-bipartitej2}, we find that
there exists a constant $t^s_k$, only depending on $k$ and $s$, such that
$Q$ has length at 
most~$t^s_k$. Hence, $G_I^4$ is $P_{t^s_k}$-free.

By construction, $G_I^4$ has a $k$-colouring if and only if $J_I^4$ has a 4-colouring that is an extension of $c_{W^4}$.
We are left to apply Lemma~\ref{l-truthbip2} and to recall that {\sc Not-All-Equal 3-SAT} with positive literals is \NP-complete.
\qed
\end{proof}

\noindent
{\bf Remark 1.}
Theorem~\ref{t-final} implies Theorem~\ref{t-coloringall}~(iii).6; we can choose $t_k=t^6_k$ for example.
We claim that a slight modification of the construction used in the proof of Theorem~\ref{t-final} gives us a better upper bound for $t_k$.
Instead of using an edge-minimal graph $F$ with chromatic number $k+1$ and girth $s$, we take the ($C_3$-free) Mycielski graph $M_{k+1}$.
Following the proof of Theorem~\ref{t-final} we pick an edge $pq$ of $M_{k+1}$ and obtain a modified graph~$M_{k+1}'$.
The graph obtained from $J_I^4$ after performing the $M_{k+1}'$-identifications is $C_3$-free.
Let $Q_1r_{i_1}Q_2r_{i_2}\cdots Q_hr_{i_h}Q_{h+1}$ be an induced path in this graph where the vertices of each $Q_i$ all belong either to an $M_{k+1}'$-copy or to $J^4_k$.
Because $|V(M_k)|=3\cdot 2^{k-2}-1$ for all $k\geq 2$ and $J^4_I$ is $P_{10}$-free, we find that, for all $k\geq 5$, $t_k\leq k + (k+1)(3\cdot 2^{k-1}-1)$.

\section{Proof of Theorem~\ref{t-coloringall}} \label{s-summary}

To prove Theorem~\ref{t-coloringall}
we need first to discuss some additional results.
Kobler and Rotics~\cite{KR03} showed that for any constants $p$ and  $k$, {\sc List $k$-Colouring} is p-time solvable on any class of graphs that have clique-width at most $p$, assuming that a $p$-expression is given.
Oum~\cite{Oum08} showed that a  $(8^{p}-1)$-expression for any $n$-vertex graph with clique-width at most $p$ can be found in $O(n^3)$ time.
Combining these two results leads to the following theorem.

\begin{theorem}\label{t-cw}
Let ${\cal G}$ be a graph class of bounded clique-width.
For all  $k\geq 1$, {\sc List $k$-Colouring}  can be solved in p-time on ${\cal G}$.
\end{theorem}

We also need the following result due to Gravier, Ho\'ang and Maffray~\cite{GHM03} who slightly improved upon a bound of
Gy\'arf\'as~\cite{Gy87} who showed that every $(K_s,P_t)$-free graph can be coloured with at most $(t-1)^{s-2}$ colours.

\begin{theorem}[\cite{GHM03}]\label{t-bound}
Let $s,t\geq 1$ be two integers.
Then every $(K_s,P_t)$-free graph can be coloured with at most $(t-2)^{s-2}$ colours.
\end{theorem}

We are now ready to prove Theorem~\ref{t-coloringall} by considering each case.
For each we either refer back to an earlier result, or give a reference; the results quoted can clearly be seen to imply the statements of the theorem.

\medskip
\noindent
{\it Proof of Theorem~\ref{t-coloringall}.}
We first consider the intractable cases of {\sc List $k$-Colouring} and note that (i).1 follows from Theorem~\ref{t-hard4}, and Theorem~\ref{t-main} implies that {\sc List $4$-Colouring} is \NP-complete for the class of $(C_5,C_6,P_6)$-free graphs which proves (i).2.
We now consider the tractable cases.
Erd\"os,  Rubin and Taylor~\cite{ERT79} and Vizing~\cite{Vi79} observed that $2$-{\sc List Colouring} is p-time solvable on general
graphs implying (i).3.
Broersma et al.~\cite{BFGP13} showed that {\sc List $3$-Colouring} is p-time solvable for $P_6$-free graphs from which we can infer (i).4 and (i).6.
Golovach et al.~\cite{GPS11} proved that for all $k,r,s,t\geq 1$,  {\sc List $k$-Colouring}  can be solved in linear time for $(K_{r,s},P_t)$-free graphs.
By taking $r=s=2$, we obtain (i).5 and (i).8.
The class of $(C_3,P_6)$-free graphs was shown to have bounded clique-width by Brandst\"adt, Klembt and Mahfud~\cite{BKM06};
using Theorem~\ref{t-cw} we see that {\sc List $k$-Colouring} is p-time solvable on $(C_3,P_6)$-free graphs for all $k\geq 1$ demonstrating (i).7.
Ho\`ang, Kami\'nski, Lozin, Sawada, and Shu~\cite{HKLSS10} proved that for all $k\geq 1$, {\sc List $k$-Colouring} is p-time solvable on $P_5$-free graphs proving (i).9.

We now consider {\sc $k$-Precolouring Extension}.
The tractable cases all follow from the results on {\sc List $k$-Colouring} just discussed.  So we are left to consider the \NP-complete cases.
Theorem~\ref{t-prebipartite} implies (ii).1 and (ii).6.
Theorems~\ref{t-new1} and~\ref{t-new2} imply (ii).4 and (ii).7, respectively.
And (ii).2, (ii).3 and (ii).5 follow immediately from corresponding results for {\sc $k$-Colouring} proved by Hell and Huang~\cite{HH13}.

Finally, we consider $k$-{\sc Colouring}; first the
\NP-complete cases. Theorem~\ref{t-new3} gives us (iii).1.
Recall that Theorem~\ref{t-final} implies Theorem~\ref{t-coloringall}~(iii).6 (see Remark~1 in Section~\ref{s-colouring}).
Hell and Huang~\cite{HH13} proved all the other \NP-completeness subcases.
Chudnovsky, Maceli and Zhong~\cite{CMZ14a,CMZ14b} announced that 3-{\sc Colouring} is p-time solvable on $P_7$-free graphs, which gives
us (iii).10 and (iii).12.
Chudnovsky, Maceli, Stacho and Zhong~\cite{CMSZ14} announced that 4-{\sc Colouring} is p-time solvable for $(C_5,P_6)$-free graphs, which
gives us (iii).15.
Theorem~\ref{t-bound} gives us (iii).17.  All other tractable cases follow from the corresponding tractable cases for {\sc List $k$-Colouring}.
\qed

\section{Open Problems}\label{s-open}

From Theorem~\ref{t-coloringall}, we see that the following cases are open in the classification of the complexity of graph colouring problems for $(C_s,P_t)$-free graphs (recall
that $t_k$ is a constant only depending on $k$).
\begin{itemize}
\item [(i)] For {\sc List $k$-Colouring} the following cases are open:
\begin{itemize}
\item [$\bullet$] $k=3$, $s=3$ and $t \geq 7$
\item [$\bullet$] $k=3$, $s\geq 5$ and $t\geq 7$
\item[$\bullet$] $k\geq 4$, $s=3$ and $t=7$.\\[-10pt]
\end{itemize}
\item[(ii)] For {\sc $k$-Precolouring Extension} the following cases are open:
\begin{itemize}
\item [$\bullet$] $k=3$, $s=3$ and $t \geq 7$
\item [$\bullet$] $k=3$, $s\geq 5$ and $t\geq 7$
\item[$\bullet$] $k=4$, $s=3$ and $7\leq t\leq 9$
\item[$\bullet$] $k=4$, $s\geq 5$ and $t=6$
\item[$\bullet$] $k=4$, $s=7$ and $t=7$
\item [$\bullet$] $k\geq 5$, $s=3$ and $7\leq t\leq 9$.\\[-10pt]
\end{itemize}
\item[(iii)] For {\sc $k$-Colouring} the following cases are open:
\begin{itemize}
\item [$\bullet$] $k=3$, $s=3$ and $t \geq 8$
\item [$\bullet$] $k=3$, $s\geq 5$ and $t\geq 8$
\item[$\bullet$] $k=4$, $s=3$ and $7\leq t\leq 21$
\item[$\bullet$] $k=4$, $s\geq 6$ and $t=6$
\item [$\bullet$] $k=4$, $s=7$ and $7\leq  t\leq 8$
\item [$\bullet$] $k\geq 5$, $s=3$ and $k+3\leq t\leq t_k-1$
\item [$\bullet$]  $k\geq 5$, $s=5$ and $t=6$.
\item[]
\end{itemize}
\end{itemize}
Besides solving these missing cases (and the missing cases from Table~\ref{t-table1}) we pose the following three problems specifically.
First, does there exist a graph $H$ and an integer $k\geq 3$ such that {\sc List $k$-Colouring} is \NP-complete and $k$-{\sc Colouring} is p-time solvable for $H$-free graphs? Theorem~\ref{t-coloringall} shows that if we forbid two induced subgraphs then the complexity of these two problems {\it can} be different: take $k=4$,
$H_1=C_5$ and $H_2=P_6$.
Second,  is {\sc List $4$-Colouring} \NP-complete for $P_7$-free bipartite graphs?
This is the only missing case of {\sc List $4$-Colouring} for $P_t$-free bipartite graphs due to Theorems~\ref{t-coloringall} and~\ref{t-hard4}.
Third, what is the computational complexity of {\sc List 3-Colouring} and {\sc 3-Precolouring Extension} for chordal bipartite graphs?

\end{document}